\documentclass[11pt]{article}

\usepackage[lined,boxed]{algorithm2e}
\usepackage{amssymb, amsmath, amsthm}
\usepackage{dirtytalk}

\usepackage{fullpage}
\usepackage{hdb_macros}
\usepackage{tikz}
\usetikzlibrary{positioning}

\newcommand{\restrict}[1]{| _{#1}}

\title{Oblivious Complexity Classes Revisited: Lower Bounds and Hierarchies}

\author{
    Karthik Gajulapalli\thanks{Georgetown University. Email: \email{kg816@georgetown.edu}.}
     \and
    Zeyong Li\thanks{National University of Singapore. Email: \email{li.zeyong@u.nus.edu}.}
 	\and
    Ilya Volkovich\thanks{Boston College. Email: \email{ilya.volkovich@bc.edu}.}
 }

\date{}
\begin{document}
\pagenumbering{roman}
\maketitle

\begin{abstract}
In this work we study \emph{oblivious} complexity classes. These classes capture the power of interactive proofs where the prover(s) are only given the input size rather than the actual input. In particular, we study the connections between the symmetric polynomial time - $\mathsf{S_2P}$ and its oblivious counterpart - $\mathsf{O_2P}$. Among our results:
\begin{itemize}
\item For each $k \in \mathbb{N}$, we construct an explicit language $L_k \in \mathsf{O_2P}$ that cannot be computed by circuits of size $n^k$.
\item We prove a hierarchy theorem for $\mathsf{O_2TIME}$. In particular, for any time constructible function $t:\mathbb{N} \rightarrow \mathbb{N}$ and any $\varepsilon > 0$ we show that: $\ensuremath{\mathsf{O_2TIME}}[t(n)] \subsetneq \ensuremath{\mathsf{O_2TIME}}[t(n)^{1 + \varepsilon}]$.

\item We prove new structural results connecting $\ensuremath{\mathsf{O_2P}}$ and $\ensuremath{\mathsf{S_2P}}$.
\item We make partial progress towards the resolution of an open question posed by Goldreich and Meir (TOCT 2015) that relates the complexity of $\mathsf{NP}$ to its oblivious counterpart - $\mathsf{ONP}$.
\item We identify a natural class of problems in $\mathsf{O_2P}$ from computational Ramsey theory, that are not expected to be in $\P$ or even $\BPP$.
\end{itemize}

To the best of our knowledge, these results constitute the first explicit fixed-polynomial lower bound and hierarchy theorem  for $\ensuremath{\mathsf{O_2P}}$. The smallest uniform complexity class for which such lower bounds were previously known was $\ensuremath{\mathsf{S_2P}}$ due to Cai (JCSS 2007). In addition, this is the first uniform hierarchy theorem for a semantic class. All previous results required some non-uniformity. 
In order to obtain some of the results in the paper, we introduce the notion of \emph{uniformly-sparse extensions} which could be of independent interest.

Our techniques build upon the de-randomization framework of the powerful Range Avoidance problem that has yielded many new interesting explicit circuit lower bounds.

\end{abstract}

\thispagestyle{empty}
\newpage

\pagenumbering{arabic}

\section{Introduction}
Proving circuit lower bounds has been one of the holy grails of theory of computation with strong connections to many fundamental questions in complexity theory.  
For instance, proving that there exists a function in $\cc{E}$\footnote{Deterministic time $2^{O(n)}$.} that requires exponential-size circuits would entail a strong derandomization: $\BPP = \P$ and $\MA = \NP$ \cite{NisanWigderson94,IW97}. And yet, while by counting arguments (i.e. \cite{Shannon49}) the vast majority of Boolean functions/languages do require exponential-size circuits, the best `explicit' lower bounds are still linear! (in fact the best known lower bound for any language in $\E^{\NP}$ is just linear \cite{li20221}).
Indeed, although it is widely \emph{believed} that $\NP$ requires super-polynomial-size circuits (i.e. $\NP \not \subseteq \PPoly$) establishing the statement even for $\NEXP$ (i.e. $\NEXP \not \subseteq \PPoly$), the exponential version of $\NP$,  has remained elusive for many years. The best known explicit lower bound is due to a seminal work of Williams \cite{Williams14}, where it was shown that $\NEXP$ requires super-polynomial-size circuits in a `very' restricted model ($\NEXP \not \subseteq \ACC^0$). 

In the high-end regime, Kannan \cite{Kannan82} has shown that the exponential hierarchy requires exponential-size circuits, via diagonalization\footnote{In fact, this argument could be viewed as solving an instance of the \emph{Range Avoidance} problem. See below.}.
More precisely, it was shown that the class $\cc{\Sigma_3E} \cap \cc{\Pi_3E}$ contains a language that cannot be computed by a circuit family of size $2^n/n$. 
This result was later improved to $\cc{\Delta_3E} = \cc{E}^{\cc{\Sigma_2P}}$ by Miltersen, Vinodchandran and Watanabe \cite{MVW99}.
Moreover, it was shown that $\cc{\Delta_3E}$ actually requires circuits of 
`maximum possible' size.  
Subsequently, the status of the problem remained stagnant for more than two decades until very recently, Chen, Hirahara and Ren \cite{chen2023symmetric} and a follow-up work by Li \cite{li2023symmetric} improved the result to $\Se$ \footnote{Symmetric exponential time.  
Indeed, 
$\Se \subseteq \cc{\Sigma_2E} \cap \cc{\Pi_2E} \subseteq \cc{\Delta_3E}.$ For a formal definition see \cref{def:St}.}.  
In particular, this result was obtained via solving the Range Avoidance ($\cc{Avoid}$) problem with `single-valued, symmetric polynomial-time' algorithm. Indeed, the focus of our work is on `oblivious' symmetric polynomial time and related complexity classes. 

\subsection{Background}

\subsubsection{Symmetric Time}

\emph{Symmetric polynomial time}, denoted by $\St$, was introduced independently by Canetti \cite{Canetti96}, and Russell and Sundaram \cite{RS98}. Intuitively speaking, this class captures the interaction between an efficient (polynomial-time) verifier $V$ and two all-powerful provers: the `YES'-prover $Y$ and the `NO'-prover $Z$, exhibiting the following behaviour: 

\begin{itemize}
    \item If $x$ is a yes-instance, then the `YES'-prover $Y$ can send an irrefutable proof/certificate $y$ to $V$ that will make $V$ \emph{accept}, \textbf{regardless} of the communication from $Z$.
    
     \item Likewise, if $x$ is a no-instance, then the `NO'-prover can send an irrefutable proof/certificate proof $z$ to $V$ that will make $V$ \emph{reject}, \textbf{regardless} of the communication from $Y$.
\end{itemize}
We stress that in both cases the irrefutable certificates can depend on $x$ itself. One can also define $\Se$ - the exponential version of $\St$, by allowing the verifier to run in linear-exponential time. For a formal definition see \cref{def:St}. A seminal result of \cite{Cai2007} provides the best known upper bound $\St \subseteq \ZPP^\NP$.
At the same time, $\St$ appears to be a very powerful class as it contains $\MA$ and $\Delta_2 \P = \P^\NP$.

\subsubsection{Oblivious Complexity Classes}

The study of oblivious complexity classes was initiated in \cite{CR06} and has subsequently received more attention  \cite{Aaronson07,FSW09,CR11,GoldreichMeir2015}. Roughly speaking, let $\Lambda$ be a complexity class such that in addition to the input $x$, the
machines $M(x,w)$ of $\Lambda$ also take a witness $w$ (and possibly other inputs). Examples of such classes include: $\NP, \MA, \St$, etc. The corresponding \emph{oblivious} version of $\Lambda$ is obtained by stipulating that the for every $n \in \N$ there exists a `common' witness $w_n$ for \textbf{all} the `respective' inputs of length $n$. For instance, a language $L$ belongs to $\ONP$ -- the oblivious version of $\NP$, if there exists a polynomial-time machine $M(x,w)$ such that:
\begin{enumerate}
    \item $\forall n \in \N$ there exists $w_n$ such that $\forall x \in \bit^{n}: x \in L \implies M(x, w_n) = 1$.

    \item $x \not \in L \implies \forall w: M(x,w) = 0$.
\end{enumerate}

Thus, in a similar manner, one can define the class $\Ot$ --- the oblivious version of $\St$, that is referred to as `oblivious symmetric polynomial time' in the literature. $\Ot$ has the additional requirement that for every $n \in \N$ there exist an irrefutable yes-certificate $y^*$ and an irrefutable no-certificate $z^*$ for all the yes-instances and the no-instances of length $n$, respectively.  
For a formal definition, see \cref{def:Ot}. 

It is immediate from the definitions that $\ONP \subseteq \NP$, $\Ot \subseteq \St$ and $\ONP \subseteq \Ot$. On the other hand, by hard-coding the witnesses/certificates we get that $\RP \subseteq \ONP \subseteq \Ot \subseteq \PPoly.$ 
In addition, it was also observed in \cite{CR06} that $\Ot$ is \emph{self-low}. That is $\Ot^{\Ot} = \Ot$.  While the oblivious classes seem to be more restricted than their non-oblivious counterparts, proving any non-trivial upper bounds could still be challenging. In terms of lower bounds, the best known containment of a non-oblivious class is $\BPP \subseteq \Ot$ \footnote{One way to see this is by observing that $\BPP \subseteq \RP^{\ONP}$ and then using the self-lowness of $\Ot$.}.
For more details and discussion see \cite{CR06,GoldreichMeir2015}.
Nonetheless, to the best of our knowledge, no ``natural" problem  for $\Ot$ (or even $\ONP$), known  to lie outside of $\BPP$, has been identified in the literature.

\subsubsection{Sparsity}

A language $L$ is \emph{sparse}, if for every input length $n \in \N$ the number of yes-instances of size $n$ is at most $\poly(n)$. We will denote the class of all sparse languages by $\SPARSE$.
Sparse languages have seen many applications in complexity theory. Perhaps, the most fundamental one is known as ``Mahaney's theorem'' \cite{Mahaney82} that implies that a sparse language cannot be $\NP$-hard, unless $\P = \NP$. In \cite{FSW09} and \cite{GoldreichMeir2015}, sparse languages were also studied in the context of oblivious complexly classes. In particular, they observed that $\NP \cap \SPARSE \subseteq \ONP$.
That is, every sparse $\NP$ language is also in $\ONP$. The same argument also implies that $\cc{NE} = \cc{ONE}$, that is, \textbf{equality} between the exponential versions of $\NP$ and $\ONP$, respectively. Given the former claim we observe that the \emph{Grid Coloring} problem, defined in \cite{AGL23}, constitutes a natural $\ONP$ (and hence $\Ot$) problem. For a formal statement, see \cref{OBS:main5}.

Subsequently, Goldreich and Meir \cite{GoldreichMeir2015} posed an open question whether a similar relation holds true for $\coNP$ and $\coONP$. That is, whether every sparse $\coNP$ language is also in $\coONP$\footnote{The original (equivalent) formulation of the question in \cite{GoldreichMeir2015} was w.r.t $\NP$ and co-sparse languages.}. Motivated by this question, we observe that essentially the same issues arise when one attempts to show that every sparse $\St$ language is also in $\Ot$. While we do not accomplish this task, we make a partial progress by introducing \emph{uniformly-sparse extensions}. The intuition behind this definition is to have a uniform `cover' of the segments of the yes-instances for \textbf{all} input lengths. For a formal definition see \cref{def:uniformly-sparse extension}. This is our main conceptual contribution.
As a corollary, we obtain that $\Se = \Oe$. Although this might not be a new result, to the best of our knowledge, this result has not appeared in the literature previously.

\subsubsection{Range Avoidance}
The study of the Range Avoidance problem ($\cc{Avoid}$) was initiated in \cite{KKMP21}. The problem itself takes an input-expanding Boolean Circuit $\cc{C}:\bit^n \rightarrow \bit^{n+1}$ as input and asks to find an element $y$, outside the range of $\cc{C}$. Since its introduction, there has been a steady line of exciting work studying the complexity and applications of $\cc{Avoid}$ \cite{korten2022hardest,GLW22,RSW22,CHLR23,GGNS23,ILW23,CGLOSS23,chen2023symmetric,li2023symmetric,CL23,KP24}.

Informally, $\cc{Avoid}$ algorithmically captures the probabilistic method where the existence of an object with some property follows from a union bound. In particular, Korten \cite{korten2022hardest} showed that solving $\cc{Avoid}$ (deterministically) would result in finding optimal explicit constructions of many important combinatorial objects, including but not limited to Ramsey graphs \cite{Radziszowski2021}, rigid matrices \cite{GLW22,GGNS23}, pseudorandom generators \cite{CT21}, two-source extractors \cite{CZ19,Li23}, linear codes \cite{GLW22}, strings with maximum time-bounded Kolmogorov complexity ($K^{\poly}$-random strings) \cite{RSW22} and truth tables of high circuit complexity \cite{korten2022hardest}.

The connection between $\cc{Avoid}$ and hard truth table makes it relevant to the study of circuit lower bounds. It has been observed and pointed out in many prior works (see, e.g. \cite{chen2023symmetric}) that proving explicit circuit lower bounds is effectively finding single-valued\footnote{Roughly speaking, a single-valued algorithm on successful executions should output a fixed (canonical) solution given the same input.} constructions of hard truth tables. Indeed, this is the framework adopted for proving circuit lower bounds in \cite{CHLR23,chen2023symmetric,li2023symmetric}: Designing a single-valued algorithm for solving $\cc{Avoid}$.

\subsubsection{Time Hierarchy Theorem}

Time Hierarchy theorems are among the most fundamental results in computational complexity theory, which (loosely speaking) assert that computation with more time is strictly more powerful. Time hierarchy theorems are known for deterministic computation ($\DTIME$) \cite{HS65,HS66} and non-deterministic computation ($\NTIME$) \cite{Cook1972,Seiferas1978,Zak83} which are syntactic classes. The situation for semantic classes such as $\BPTIME$ is much more elusive as it is unclear how to enumerate and simulate all $\BPTIME$ machines while ensuring that the simulating machine itself remains a proper $\BPTIME$ machine. In fact, even verifying that a machine is a $\BPTIME$ machine is itself an undecidable problem. For $\BPTIME$, a time hierarchy theorem is only known for its promise version, or when given one bit of advice \cite{Barak02,FS04,FST05}. This was further generalized in \cite{vMP07}, where they show most semantic classes (e.g. \MA, \Stt) admit a time hierarchy theorem with one bit of advice.

Along a different line of research, it was shown in \cite{LOS21, DPWV22} that coming up with a pseudo-deterministic algorithm (single-valued randomized algorithms) for estimating the acceptance probability of a circuit would imply a uniform hierarchy theorem for $\BPTIME$.

\subsection{Previous Results}

A parallel line of work focused on the `low-end' regime by proving the so-called `fixed-polynomial' circuit lower bounds. That is, the goal is to show that for every $k \in \N$ there is a language $L_k$ (that may depend on $k$) which cannot be computed by circuits of size $n^k$.
The first result in this sequel --- fixed-polynomial lower bounds for the polynomial hierarchy,
was obtained by Kannan \cite{Kannan82} via diagonalization. In particular, it was shown that for every $k \in \N$ there exists a language  $L_k \in \cc{\Sigma_4 \P}$ that cannot be computed by circuits of size $n^k$. This result was then improved to $\cc{\Sigma_2P}$.  The key idea behind this and, in fact, the vast majority of subsequent improvements is a `win-win' argument that relies on the \emph{Karp-Lipton} collapse theorem \cite{KarpLipton80}:
if $\NP$ has polynomial-size circuits (i.e $\NP \subseteq \PPoly$) then the (whole) polynomial hierarchy collapses to $\cc{\Sigma_2P}$. More specifically, the argument proceeds by a two-pronged approach:
\begin{itemize}
    \item Suppose $\NP \not \subseteq \PPoly$. Then the claim follows as $\NP \subseteq \cc{\Sigma_2P}$.
    \item On the other hand, suppose $\NP \subseteq \PPoly$. Then by Karp-Lipton:  $\cc{\Sigma_4P} =  \cc{\Sigma_2P}$ and in particular for all $k \in \N: L_k \in \cc{\Sigma_2P}$.
\end{itemize}
Indeed, by deepening the collapse, the result was further improved to $\ZPP^\NP$ \cite{KW98,BCGKT96}, $\P^{\cc{pr}\MA}$ \cite{CR11} and $\St$ \cite{Cai2007}. By using different versions of the Karp-Lipton theorem, the result has also been extended to  $\cc{PP}$ \cite{Vinodchandran05,Aaronson06} and $\MA/1$  \cite{Santhanam09}. 

Yet, despite the success of the `win-win' argument, the obtained lower bounds are often non-explicit due to the non-constructiveness nature of the argument. Different results \cite{CW04,Santhanam09} were required to exhibit explicit `hard' languages in $\cc{\Sigma_2P}, \PP$ and $\MA/1$. Nonetheless, the last word about $\St$ is yet to be said. For instance, we know that there is a language in $\St$ that requires circuits of size, say, $n^2$ from such arguments. However, prior to the results of \cite{chen2023symmetric,li2023symmetric}, one could not prove any super-linear lower bound for \textbf{any} particular language in $\St$. Another limitation of the `win-win' argument stems from the fact that it only applies to complexity classes which (provably) contain $\NP$. In particular, in \cite{CR06} it was actually shown that if $\NP \subseteq \PPoly$ then the polynomial hierarchy collapses all the way to $\Ot$!
Unfortunately, this result does not immediately imply fixed-polynomial lower bounds for $\Ot$\footnote{Indeed, the authors in \cite{CR06} could only obtain fixed-polynomial lower bounds for $\NP^{\Ot}$ which was later subsumed by the results of \cite{Santhanam09}.}
as it is unknown and, in fact, \emph{unlikely} that $\Ot$ contains $\NP$. Furthermore, such a containment will be `self-defeating'. Recall that $\Ot \subseteq \PPoly$. Hence, if $\NP \subseteq \Ot$ then $\NP \subseteq \PPoly$ which in and of itself already implies the collapse of the whole polynomial hierarchy! 
   
Finally, it is important to mention a result of \cite{FSW09} that for any $k \in \N$, $\NP$ has circuits of size $n^k$ iff $\ONP/1$ does. In that sense $\ONP$ already nearly captures the hardness of showing fixed-polynomial lower bounds for $\NP$.

\subsection{Our Results}

In our first result we extend the lower bounds for $\St$ and $\Se$, to their weaker oblivious counterparts $\Ot$ and $\Oe$, respectively. This result follows the recent line of research that obtains circuit lower bounds by means of \emph{deterministically} solving (i.e. derandomizing) instances of the Range Avoidance problem \cite{CHLR23,chen2023symmetric,li2023symmetric}.

\begin{maintheorem}
\label{THM:main1}
 For all $k \in \N$, $\Ot   \not \subseteq  \SIZE[n^k]$. Moreover, for each $k$ there exists an explicit language $L_k \in \Ot$ such that $L_k \not \in \SIZE[n^k]$.
\end{maintheorem}

In fact we prove a stronger parameterized version of this statement (see \cref{thm:O2_LB}, \cref{cor:OP_SP}, and \cref{cor:fixed_poly_direct}). We now highlight three main reasons why such a result is fascinating:
\begin{enumerate}
    \item Our lower bound does not follow the framework of \say{win-win} style Karp-Lipton collapses. As was mentioned above,  since already $\Ot \subseteq \PPoly$ the pre-requisite for proving the bound via the \say{win-win} argument will be self-defeating. 

    \item Our proof is constructive and for every $k \in \N$ we define an explicit language $L_k \in \Ot$ for which $L_k \not\subseteq \cc{SIZE}[n^k]$.

    \item $\Ot$ becomes the smallest uniform complexity class known for which fixed-polynomial lower bounds are known. Moreover, after more than 15 years, this class coincides again with the deepest known collapse result of the Karp-Lipton Theorem\footnote{Indeed, in the universe of  \cite{Cai2007} and \cite{CR06} prior to our work, the smallest class has been $\St$, while the deepest known collapse was to $\Ot$.}.
 
\end{enumerate}

Our second result gives a hierarchy theorem for $\Ott$.

\begin{maintheorem}
\label{THM:main2}
For any time constructible function $t:\N \rightarrow \N$ such that $t(n) \geq n$ and any $\varepsilon > 0$ 
 it holds that: $\Ott[t(n)] \subsetneq \Ott[t(n)^{1 + \varepsilon}]$.
\end{maintheorem}

We remark, that to the best of our knowledge, this is the first known hierarchy theorem for a uniform semantic class (that contains $\BPTIME$). At the same time, we observe that the proof of the non-deterministic time hierarchy theorem  ($\cc{NTIME}$) (see e.g. \cite{Zak83}) actually extends to the \emph{oblivious} non-deterministic time ($\cc{ONTIME}$) since the hard language constructed in their proof is unary and hence is contained in $\cc{ONTIME}$. On the other hand, that same language also diagonalizes against \textbf{all} $\cc{NTIME}$ machines which is a superset of all $\cc{ONTIME}$ machines.

In our time hierarchy theorem for $\Ott$, which goes through a reduction to $\cc{Avoid}$, one can view $\cc{Avoid}$ as a tool for diagonalization against all circuits of fixed size, which in turn contains all $\Ott$ machines with bounded time complexity. This (together with the time hierarchy theorem for $\cc{ONTIME}$) might suggest an approach for proving time hierarchy theorem for semantic classes in general: diagonalize against a syntactic class that encompasses the semantic class in consideration.\\

Finally, we introduce the notion of $\emph{uniformly-sparse extensions}$ (for a formal definition, see \cref{def:uniformly-sparse extension}) to get structural complexity results relating $\Ott$ and $\Stt$. This relation provides an alternate method of proving \cref{THM:main1}. 

\begin{maintheorem}
\label{THM:main3}
Let $L \in \St$. If $L$ has a \emph{uniformly-sparse extension} then  $L \in \Ot$.
\end{maintheorem}

While not much was known between the classes $\Ott$ and $\Stt$, except that $\Ott \subseteq \Stt$, we show new connections between the two classes. In fact, we prove a stronger parameterized version of \cref{THM:main3} that yields as a corollary a proof of the equivalence $\Se = \Oe$ (see \cref{cor:OE_SE}). Going back to the original motivation, by repeating the same argument, we make a partial progress towards the resolution of the open question posed by Goldreich and Meir in \cite{GoldreichMeir2015}.
See \cref{lem:Ot_simul_St} for more details.

\begin{maintheorem}
\label{THM:main4}
    Let $L \in \coNP$. If $L$ has a \emph{uniformly-sparse extension} then $L \in \coONP$.
\end{maintheorem}

Finally, we observe that computational Ramsey theory provides some very natural problems in $\ONP$ (and hence \Ot). As an example, we re-introduce the grid coloring problem below. While Claim 2.5 in \cite{GoldreichMeir2015} suggests a generic way to generate problems via padding arguments\footnote{The approach is to pick a language $L$ in $\Se$ that is (conjectured) not in $\cc{BPE}$. Then the padded version of $L$ will be in $\Ot \setminus \BPP$.}, these problems are, however, not very intuitive. 

\begin{definition}[Grid Coloring \cite{AGL23}]\; \;\\
$\mathsf{GC} = \condset{(1^n01^m01^c)}{\text{the $n \times m$ grid can be $c$-colored and not have any monochromatic squares.}}$
\end{definition}

Note that Grid Coloring is an example of one of such problems that are in $\NP \cap \SPARSE \subseteq \ONP$, and hence unlikely to be in $\BPP$. Other problems that come from computational Ramsey theory like the \emph{Gallai-Witt} theorem, and the \emph{Van der Waerden's} theorem have a very similar flavor.

\begin{mainobservation}
 \label{OBS:main5}
$\mathsf{GC} \in \ONP \subseteq \Ot$.
\end{mainobservation}

\noindent Below we make a few remarks. For a further discussion see \cite{Gasarch2010}.
\begin{itemize}
    \item $\mathsf{GC} \in \NP$ since the coloring itself is a witness.
    \item $\mathsf{GC}$ is not known to be in $\P$ or even $\BPP$.
    \item $\mathsf{GC} \in \SPARSE$. In fact, $\mathsf{GC}$ has a uniformly-sparse extension.
    \item Therefore, by the results of \cite{FSW09,GoldreichMeir2015}, $\mathsf{GC} \in \ONP$.
    \item On the other hand, by Mahaney's theorem $\mathsf{GC}$ is \emph{unlikely} to be $\NP$-complete.
\end{itemize}

\subsection{Proof Overview}

Our work builds on the recent line of work on Range Avoidance. \cite{korten2022hardest} provides a reduction of generating hard truth tables from $\cc{Avoid}$, and \cite{chen2023symmetric, li2023symmetric} give a single-valued $\St$ time algorithm for $\cc{Avoid}$.

\paragraph{Avoid Framework for Circuit Lower bounds} Let $\cc{TT_{n,s}}: \bit^{s \log s} \rightarrow \bit^{2^n}$ be the truth table generator circuit (see \cref{def:TT}), i.e. $\cc{TT_{n,s}}$ take as input an encoding of a $n$-input $s$-size circuit and outputs the truth table of the circuit. By construction, $\cc{TT_{n,s}}$ maps all circuits of size $s$ (encoded using $s \log s$ bits) to their corresponding truth tables. Then, $\cc{Avoid(TT_{n,s})}$ will output a truth-table not in the range of $\cc{TT_{n,s}}$ and hence not decided by any $s$-sized circuit (a circuit lower bound!!). For correctness we only need to ensure that $s \log s < 2^n$, so the $\cc{TT_{n,s}}$ is input-expanding, and hence a valid instance of $\cc{Avoid}$.

While the above construction gives us a way of getting explicit exponential lower bounds against even $\cc{SIZE}[2^n/n]$, the input to \cc{Avoid} is also exponential in input length $n$. As a result, the lower bounds we get are for the exponential class $\Se$ and not $\St$. Note that one can scale down this lower bound in a black-box manner to get fixed-polynomial lower bounds for $\St$, but will lose explicitness in the process.

To fix this we modify the above reduction to take as input the prefix truth table generator circuit, $\cc{PTT_{n,s}}: \bit^{s \log s} \rightarrow \bit^{s \log s + 1}$, where instead of evaluating the input circuit on the whole truth table, $\cc{PTT_{n,s}}$ evaluates on the lexicographically first $(s \log s + 1)$ inputs (see \cref{def:PTT}). Let $f_{n,s} = \cc{Avoid(PTT_{n,s})}$, and define the truth table of the hard language to be $L := f_{n,s} || 0^{2^n - s \log s -1}$. By construction, $L$ cannot be decided by any $n$-input $s$-size circuit. Moreover, when $s$ is polynomial, the size of  $\cc{PTT_{n,s}}$ is also polynomial\footnote{In literature the complexity of computing $\cc{PTT_{n,s}}$ (Circuit-Eval) is often left as $\poly$, however for our application of getting explicit lower bounds it is crucial to get its fine-grained complexity (see \cref{lem:tm_evaluate_ck} and \cref{lem:ptt_size}).} (\cref{lem:ptt_size}). Hence the single-valued\footnote{For the language to be well defined it is essential for the output of our algorithm to be single-valued.} algorithm computing $f_{n,s}$ is in $\St$ and the explicit fixed-polynomial bounds follow.

To see that the language $L \in \Ot$, observe that the $\St$ time algorithm is oblivious to $x$, since for any $x$ of length $n$, $f_{n,s}$ is the same. One important observation here is that, for the purpose of obtaining circuit lower bound, it suffices to solve Range Avoidance on \emph{one} specific family of circuits (the truth table generating circuit that maps another circuit to its truth table). Hence, while it is unclear whether Range Avoidance can be solved in $\cc{FO_2P}$, we could still obtain circuit lower bound for $\Ot$.

\paragraph{Hierarchy Theorems for $\Ott$} To get a hierarchy theorem for $\Ott$, we first get an upper bound on $\Ott$ computation via a standard Cook-Levin argument that converts the $\Ott$ verifier into a circuit (SAT-formula) for which we can hard code the \say{YES} and \say{NO} irrefutable certificates at every input length (\cref{lem:O2_UB}). A lower bound follows via the $\cc{Avoid}$ framework discussed above (\cref{thm:O2_LB}). We can now lift the hierarchy theorem on circuit size (\cref{thm:ckt_size_hierarchy}) to get a hierarchy on $\Ott$ (see \cref{thm:O2_heirarchy}).

\paragraph{Sparsity and Lower Bounds} We begin by introducing the notion of \emph{uniformly-sparse extensions}. Roughly speaking a sparse language $L$ has a \emph{uniformly-sparse extension} if there is a language $L' \in \P$, such that $L \subseteq L'$ and $L'$ is also sparse (for formal definitions see \cref{sec:prelim_sparsity}). 

We show that if a language $L \in \St$ has a \emph{uniformly-sparse extension}, then $L \in \Ot$. Let $L'$ be the \emph{uniformly-sparse extension} of a language $L \in \St$ and let $X = \{x \in L'\}$. Since $L' \in \P$, we first apply the polynomial time algorithm for $L'$ which let us filter out most inputs, i.e. $x \notin L'$, and hence $x \notin L$. We are now left with deciding membership in $L$ over the set $X$, where $|X| \leq \poly$. 

Let $V^*$ be the polynomial time $\cc{S_2}$-verifier for $L$, then for every $x \in X$ there exists either an irrefutable YES certificate ($y_x$) s.t. $V^*(x, y_x, \cdot) = 1$, or an irrefutable NO certificate ($z_x$) s.t. $V^*(x, \cdot , z_x) = 0$. Let $Y$ be the set of all such $y_x$'s and $Z$ be the set of all such $z_x$'s. Now for any $x \in X$, it suffices to find the correct $y_x$ from $Y$ (or $z_x$ from $Z$) and apply $V^*(x, y_x, z_x)$ to decide $x$.

In \cref{lem:Ot_simul_St} we prove a more efficient parameterized version of this argument. In addition, we are able to apply this in the exponential regime to show the equivalence $\Oe = \Se$ (see \cref{cor:OE_SE}).

\section{Preliminaries}

Let $L \subseteq \bit^*$ be a language. For $n \geq 1$ we define the \emph{$n$-th slice of $L$}, 
$L \restrict{n} := L \cap \bit^n$ as all the strings in $L$ of length $n$. The characteristic string of $L \restrict {n}$, denoted by $\mathcal{X}_{L \restrict{n}}$, is the binary string of length $2^n$ which represents the truth table defined by $L \restrict{n}$.

\subsection{Complexity Classes}
We assume familiarity with complexity theory and notion of non-uniform circuit families (see for e.g. \cite{AroraBarak09, Goldreich08}).

\begin{definition}[Deterministic Time]
    Let $t: \N \rightarrow \N$. We say that a language $L \in \cc{TIME}[t(n)]$, if there exists a deterministic time multi-tape Turing machine that decides $L$, in at most $O(t(n))$ steps.
\end{definition}

\begin{definition}[Symmetric Time]
\label{def:St}
    Let $t: \N \rightarrow \N$. We say that a language $L \in \Stt[t(n)]$, if there exists a $O(t(n))$-time predicate $P(x,y,z)$ that takes $x\in \bit^n$ and $y,z\in\bit^{t(n)}$ as input, satisfying that:
    \begin{enumerate}
        \item If  $x \in L$ , then there exists a $y$ such that for all $z$, $P(x,y,z) = 1$.
        \item If $x \notin L$, then there exists a $z$ such that for all $y$, $P(x,y,z) = 0$. 
    \end{enumerate}

    \noindent Moreover, we say $L \in \St$, if $L \in \Stt[p(n)]$ for some polynomial $p(n)$, and
    $L \in \Se$, if $L \in \Stt[t(n)]$ for $t(n) \leq 2^{O(n)}$.
\end{definition}

\begin{definition}[Single-valued $\FSt$ algorithm]
    A single-valued $\FSt$ algorithm $A$ is specified by a polynomial $\ell(\cdot)$ together with a polynomial-time algorithm $V_A(x, \pi_1, \pi_2)$. On an input $x\in\bit^*$, we say that $A$ outputs $y_x \in \bit^*$, if the following hold:
    \begin{enumerate}
        \item There exists a $\pi_1 \in \bit^{\ell(|x|)}$ such that for every $\pi_2 \in \bit^{\ell(|x|)}$, $V_A(x,\pi_1, \pi_2)$ outputs $y_x$.
        \item For every $\pi_1 \in \bit^{\ell(|x|)}$ there exists a $\pi_2 \in \bit^{\ell(|x|)}$, such that $V_A(x,\pi_1, \pi_2)$ outputs either $y_x$ or $\bot$.
    \end{enumerate}
    And we say that $A$ solves a search problem $\Pi$ if on any input $x$ it outputs a string $y_x$ and $y_x \in \Pi_x$, where a search problem $\Pi$ maps every input $x \in \bit^*$ into a solution set $\Pi_x \subseteq \bit^*$.
    
\end{definition}

We now formally define $\Ott$ - the oblivious version of the class $\Stt$.
The key difference is that unlike $\Stt$, where each irrefutable yes/no certificate can depend on the input $x$ itself, in $\Ott$ the yes/no certificates can \textbf{only} depend on $\size{x}$, the length of $x$. In other words, for every input length $n$, there exist a common YES-certificate $\mathbf{y^*}$ and a  common NO-certificate $\mathbf{z^*}$ for checking membership of $x \in L \restrict{n}$.

\begin{definition}[Oblivious Symmetric Time]
\label{def:Ot}
  Let $t: \N \rightarrow \N$. We say that a language $L \in \Ott[t(n)]$, if there exists a $O(t(n))$-time predicate $P(x,y,z)$ such that for every $n \in \N$ there exist $\mathbf{y^*}$ and $\mathbf{z^*}$ of length $O(t(n))$ satisfying the following for every input $x \in \bit^n:$ 
    \begin{enumerate}
        \item If $x \in L$, then for all $z$, $P(x,\mathbf{y^*},z) = 1.$
        \item If $x \notin L$, then for all $y$, $P(x,y,\mathbf{z^*}) = 0$. 
    \end{enumerate}
        
   \noindent Moreover, we say $L \in \Ot$, if $L \in \Ott[p(n)]$ for some polynomial $p(n)$, and
    $L \in \Oe$, if $L \in \Ott[t(n)]$ for $t(n) \leq 2^{O(n)}$.
\end{definition}

\subsection{Nonuniformity}
 We recall certain circuit properties:

\begin{definition}
    A boolean circuit $C$ with $n$ inputs and size $s$, is a Directed Acyclic Graph (DAG) with $(s + n)$ nodes. There are $n$ source nodes corresponding to the inputs labelled $1, \hdots, n$ and one sink node labelled $(n+s)$ corresponding to the output. Each node, labelled $(n+i)$, for $1 \leq i \leq s$ has an in-degree of $2$ and corresponds to a gate computing a binary operation over its two incoming edges.
\end{definition}

\begin{definition}
    Let $s: \N \rightarrow \N$. We say that a language $L \in \cc{SIZE}[s(n)]$ if $L$ can be computed by circuit families of size $O(s(n))$ for all sufficiently large input size $n$.
\end{definition}

\begin{definition}
    Let $s: \N \rightarrow \N$. We say that a language $L \in i.o.\text{-}\cc{SIZE}[s(n)]$ if $L$ can be computed by circuit families of size $O(s(n))$ for infinitely many input size $n$.
\end{definition}

By definition, we have $\cc{SIZE}[s(n)] \subseteq i.o.\text{-}\cc{SIZE}[s(n)]$. Hence, circuit lower bounds against $i.o.\text{-}\cc{SIZE}[s(n)]$ are stronger and sometimes denoted as \emph{almost-everywhere} circuit lower bound in the literature.

We now state the hierarchy theorem for circuit size. The standard proof of this result is existential and goes through a counting argument (see e.g. \cite{AroraBarak09}). However, we comment that using the framework of $\cc{Avoid}$, we can now actually get a constructive size hierarchy theorem, albeit with worse parameters.

\begin{theorem}[Circuit Size Hierarchy Theorem\cite{AroraBarak09}]
\label{thm:ckt_size_hierarchy}
    For all functions $s: \N \rightarrow \N$ with $n \leq s(n) < o(2^n/n)$: $$\cc{SIZE}[s(n)] \subsetneq \cc{SIZE}[10s(n)]\; .$$
\end{theorem}

For our applications, it will be essential to have a tight encoding scheme for circuits. In fact, we will also need the fine-grained complexity of the Turing machine computing Circuit-Eval (i.e. given as input a description of a circuit $C$ and a point $x$, computes $C(x)$).

\begin{lemma} \label{lem:ckt_encoding}
    For $n,s \in \N$, and $s \geq n \geq 12$ , any $n$-input, $s$-size circuit $C$, there exists an encoding scheme $E_{n,s}$ which encodes $C$ using $5 s \log s$ bits.
\end{lemma} 

\begin{proof}
    Let $C$ be an $n$-input, $s$-size circuit, we now define $E_{n,s}$. Each gate label from $1, \hdots, n+s$ can be encoded using $\log(n + s)$ bits. First encode the $n$ inputs using $n \log(n + s)$ bits. Next fix a topological ordering of the remaining gates. For each gate we can encode its two inputs (two previous gates) with $2 \log(n + s)$ bits and the binary operation which requires $4$ bits (since there 16 possible binary operations). So the length of our encoding is $n \log (n + s) + s (2 \log (n + s) + 4) \leq 3s\log(2s) + 4s \leq  5 s \log s$ for all $n \geq 12$.
\end{proof}

\begin{lemma} \label{lem:tm_evaluate_ck}
    For $n,s \in \N$, and $s \geq n$, let $E_{n,s}$ be an encoding of an $n$-input, $s$-size circuit $C$ using \cref{lem:ckt_encoding}. Then there exists a multi-tape Turing machine $M$ such that $M(E_{n,s}, x) = C(x)$ and it runs in $O(s^2 \log s)$ time.
\end{lemma}

\begin{proof}
    We utilize one tape (memory tape)  to store all the intermediate values computed at each gate $g_i$ using $n + s$ cells, and a second tape (evaluation tape) using $6$ cells to compute the value at each $g_i$. We process each gate sequentially as it appears in the encoding scheme, and let $g_{i_l}$ and $ g_{i_r}$ be the two gates feeding into $g_i$. Since \cref{lem:ckt_encoding} encodes the gates in a topological order, we can assume that when computing $g_i$, both $g_{i_l}$ and $g_{i_r}$ have already been computed.
    First copy the value of input bits of $x$ onto the memory tape, and move the head of the input tape to the right by $n \log(n + s)$ steps in $O(s \log s)$ time. Now to compute a gate $g_i$ we write the values of $g_{i_l}$ and $g_{i_r}$ along with the binary operation onto the evaluation tape. We can compute any binary operation with just constant overhead and write its value onto the $i$th cell of the memory tape. To output the evaluation of the circuit we output the value on the $(n + s)$th cell of the memory tape. The cost of evaluating each gate is dominated by the 2 read and 1 write operations on the memory tape that take $O(s)$ time. Since the size of the input upper bounds the number of gates we have that the simulation takes $O(s|E_{n,s}|) = O(s^2 \log s)$ time.
\end{proof}

Finally, we recall the famous Cook-Levin theorem that lets us convert a machine $M \in \cc{TIME}[t(n)]$ into a circuit $C \in \cc{SIZE}[t(n) \log t(n)]$. 

\begin{theorem}[Cook-Levin Theorem \cite{AroraBarak09}]\label{lem:tm_to_ckt}
    Let $t: \N \rightarrow \N$ be a time constructible function. Then any multi-tape Turing machine running in $\cc{TIME}[t(n)]$ time can be simulated by a circuit-family of $\cc{SIZE}[t(n) \log t(n)]$.
\end{theorem}

\subsection{Range Avoidance}

\begin{definition}
    The Range Avoidance ($\cc{Avoid}$) problem is defined as follows: given as input the description of a Boolean circuit $C:\bit^n \rightarrow \bit^{m}$, for $m > n$, find a $y \in \{0, 1\}^m$ such that $\forall x \in \{0, 1\}^n: C(x) \neq y$.
\end{definition}

An important object that connects $\cc{Avoid}$ and circuit lower bound is the truth table generator circuit.

\begin{definition}\label{def:TT}\cite[Section 2.3]{chen2023symmetric}
     For $n, s\in \N$ where $n \leq s \leq 2^n$, the truth table generator circuit $\cc{TT}_{n,s}:\bit^{L_{n,s}} \rightarrow \bit^{2^n}$ maps a $n$-input size $s$ circuit using $L_{n,s} = (s+1)(7+\log(n+s))$ bits of description\footnote{in fact, it maps a stack program of description size $L_{n,s}$ and it is known that every $n$-input size $s$ circuit has an equivalent stack program of size $L_{n,s}$ \cite{frandsen2005reviewing}.} into its truth table. Moreover, such circuit can be uniformly constructed in time $\poly(2^n)$.
\end{definition}

For the purpose of obtaining fixed polynomial circuit lower bound, we generalise the truth table generator circuit above into one that outputs the prefix of the truth table. We also use a different encoding scheme (with constant factor loss in the parameter) for the convenience of presentation.

\begin{definition}\label{def:PTT}
     For $n, s\in \N$ where $12 \leq n \leq s \leq 2^n$ and $|E_{n,s}| = 5s \log s < 2^n$, the prefix truth table generator circuit $\cc{PTT}_{n,s}:\bit^{|E_{n,s}|} \rightarrow \bit^{|E_{n,s}|+1}$ maps a $n$-input circuit of size $s$ described with $|E_{n,s}|$ bits into the lexicographically first $|E_{n,s}| + 1$ entries of its truth table. %
\end{definition}

Since we want to prove lower bounds not just in the exponential regime, but also in the polynomial regime for any fixed polynomial, we need a more fine-grained analysis for the running time of uniformly generating $\cc{PTT}_{n,s}$

\begin{lemma}\label{lem:ptt_size}
    The prefix truth table generator circuit $\cc{PTT}_{n,s}:\bit^{|E_{n,s}|} \rightarrow \bit^{|E_{n,s}|+1}$ has size $O(|E_{n,s}|^3)$ and can be uniformly constructed in time $O(|E_{n,s}|^3)$.
\end{lemma}

\begin{proof}
     Let $M$ be the multi-tape Turing machine from \cref{lem:tm_evaluate_ck} that takes as input an encoding of a circuit and a bitstring, and evaluates the circuit on that bitstring. Let $C$ be the circuit generated from \cref{lem:tm_to_ckt} that simulates $M$. Then $\cc{SIZE}(C) = O(s^2 \log^2 s) = O(|E_{n,s}|^2)$. Making $|E_{n,s}| + 1$ copies of $C$ for each output gate gives a circuit of size $O(|E_{n,s}|^3)$.
    
\end{proof}

\begin{theorem}[\cite{li2023symmetric,chen2023symmetric}]\label{Alg:FS2P}
    There exists a single-valued $\cc{F}\St$ algorithm for $\cc{Avoid}$. Moreover, on input circuit $C:\bit^n \rightarrow \bit^{n+1}$, the algorithm runs in time $O(n|C|)$\footnote{the running time was implicit in the proof of \cite{li2023symmetric}, but easy to verify.}.
\end{theorem}

\begin{theorem}[\cite{li2023symmetric,chen2023symmetric}]\label{thm:S2E_LB}
    There exists an explicit language $L \in \Se \setminus i.o.\text{-}\cc{SIZE}[2^n/n]$.
\end{theorem}
\begin{proof}
    For any $n \in \Z$, let $\cc{TT}_n:\bit^{2^n-1} \rightarrow \bit^{2^n}$ be the truth table generator circuit. Let $f_n \in \bit^{2^n}$ be the canonical solution output by the single-valued algorithm from \cref{Alg:FS2P} on input $\cc{TT}_n$. 
    
    The hard language $L$ is defined as follows: for any $x \in \bit^*$, $x \in L$ if and only if the $(x+1)$-th bit of $f_{|x|} = 1$, treating $x$ as an integer from $0$ to $2^n-1$.
\end{proof}

\subsection{Sparse Languages}\label{sec:prelim_sparsity}
We define some notions of sparsity below, we first introduce natural definitions of sparsity and \emph{sparse extensions} in the polynomial regime, and then give their generalizations in the fine-grained setting.

\begin{definition}
    A language $L \in \cc{SPARSE}$ if for all $n$, $\left| L \cap \bit^n \right| \leq \cc{poly}(n)$. Moreover, $L$ is called \emph{uniformly-sparse} if $L \in \P \cap \cc{SPARSE}$.
\end{definition}

It is easy to see that $\cc{SPARSE} \subseteq \PPoly$. That is, one can identify the yes-instances efficiently, albeit in non-uniform fashion. The purpose of introducing the \emph{uniform-sparsity} is to be able to identify these inputs efficiently in a uniform fashion. Unfortunately, we cannot expect any such language $L$ to lie even in a modestly hard class as, by definition, $L \in \P$. The purpose of the \emph{uniformly-sparse extensions}, on the other hand, is to bridge this gap.
One can observe that unlike the \emph{uniformly-sparse} languages, which are contained in $\P$, languages with uniformly-sparse extension can even be undecidable! In particular, any unary language has uniformly-sparse extension in form of $1^*$.

\begin{definition}
\label{def:uniformly-sparse extension}
    A language $L$ has a \emph{uniformly-sparse extension}, if there exists a $L'$ s.t. :
        \begin{enumerate}
            \item $L \subseteq L'$
            \item $L'$ is \emph{uniformly-sparse}
        \end{enumerate}
\end{definition}

Generalizing the above definitions in the fine-grained setting, we get:

\begin{definition}
    Let $t : \N \rightarrow \N$ be a computable function, then a language $L$ is $t(n)$-$\cc{SPARSE}$ if for all $n,\left| L \cap \bit^n \right| = O(t(n))$. Moreover we say that $L$ is \emph{$t(n)$-uniformly-sparse} if $L \in \cc{TIME}[t(n)] \cap t(n)$-$\cc{SPARSE}$.
\end{definition}

\begin{definition}
    $L$ has a \emph{$t(n)$-uniformly-sparse extension}, if there exists a $L'$ s.t.:
    \begin{enumerate}
        \item $L \subseteq L'$
        \item $L'$ is \emph{$t(n)$-uniformly-sparse}.
    \end{enumerate}
\end{definition}

\noindent Observe that \textbf{every} binary language $L$ is $2^n$-$\cc{SPARSE}$. Furthermore, every such $L$ has a trivial $2^n$-uniformly-sparse extension: $\bit^*$.

\section{Lower Bounds \& Hierarchy Theorem }

In this section, we first present (\cref{thm:S2_LB})  a fine-grained, parameterised version of \cref{thm:S2E_LB}. This allows us to use the $\cc{Avoid}$ framework and get circuit lower bounds in $\Stt[t(n)]$ instead of $\Se$. We then observe that our $\Stt[t(n)]$ witness is oblivious of the input, and hence the lower bounds we get are actually in $\Ott[t(n)]$ as highlighted in \cref{thm:O2_LB}. 

In \cref{thm:O2_heirarchy} and \cref{thm:O2_heirarchy_fine_grained} we present the first time hierarchy theorem for $\Ot$. In fact, we note to the best of our knowledge that this is the first known time hierarchy theorem for a semantic class.

\begin{theorem}\label{thm:S2_LB}
    For $n \in \N$, let $t:\N \rightarrow \N$ be a time-constructible function, s.t. $t(n) > n \geq 12$ then  $$\Stt[t(n)]\not\subseteq i.o.\text{-}\SIZE\left[\frac{t(n)^{1/4}}{\log(t(n))}\right] \; .$$
\end{theorem}

\begin{proof}
    We construct a language $L_t \in \Stt[t(n)]$ and $L_t \not\subseteq i.o.\text{-}\SIZE\left[\frac{t(n)^{1/4}}{\log(t(n))}\right]$.
    
    For any $n \in \N$, let $s = \floor{\frac{t(n)^{1/4}}{\log (t(n))}}$ and $|E_{n,s}| = \ceil{5s \log s}$. Let $\cc{PTT}_{n,s}:\bit^{|E_{n,s}|} \rightarrow \bit^{|E_{n,s}|+1}$ be the prefix truth table generator circuit as in \cref{def:PTT}. Let $f_{n} \in \bit^{|E_{n,s}|+1}$ be the canonical solution to $\cc{Avoid}(\cc{PTT_{n,s}})$ as outputted by the single-valued algorithm from \cref{Alg:FS2P}. 
    
    The hard language $L_t$ is defined as follows: for any $n \in \Z$, the characteristic string of $L_t\restrict{n}$ is set to be $\mathcal{X}_{L_t\restrict{n}} := f_n || 0^{2^n - |E_{n,s}| - 1}$.

    By definition of $\cc{PTT}_{n,s}$ and the fact that $f_n \notin \cc{Image}(\cc{PTT}_{n,s})$, we have that $L_t \notin i.o.\text{-}\SIZE\left[ s \right]$. On the other hand, the single-valued algorithm for finding $f_n$ runs in time $O(|E_{n,s}| \cdot |\cc{PTT}_{n,s}|) = O(t(n))$. Hence, $L_t\in \Stt[t(n)]$. 
\end{proof}

We make the observation that the witness in the $\Stt$ machine above is oblivious to the actual input $x$.
\begin{theorem}\label{thm:O2_LB}
    For $n \in \N$, let $t:\N \rightarrow \N$ be a time-constructible function, s.t. $t(n) > n \geq 12$ then 
    $$\Ott[t(n)]\not\subseteq i.o.\text{-}\SIZE\left[\frac{t(n)^{1/4}}{\log(t(n))}\right].$$
\end{theorem}
\begin{proof}
    Consider the same language $L_t$ in the proof of \cref{thm:S2_LB}. Notice that for any input $x$ of the same length $n$, the $\cc{F}\St$ algorithm is run on the same circuit $\cc{PTT}_{n,s}$ and hence the witness is the same for inputs of the same length. Thus, it follows that $L_t \in \Ott[t(n)]$.
\end{proof}

We now get as a corollary a proof of \cref{THM:main1}.

\begin{corollary}\label{cor:fixed_poly_direct}
    For all $k \in \N$, there exists an explicit language $L_k \in \Ot$ s.t. $L_k \not\subseteq \cc{SIZE}[n^k]$.
\end{corollary}

\begin{proof}
    Fix $t(n) = n^{5k}$. Then there is an explicit hard language $L_t$ as defined in the proof of \cref{thm:S2_LB}, such that $L_t \not\subseteq \cc{SIZE}[n^k]$. Moreover, by \cref{thm:O2_LB} we have that $L_t \in \Ot$.
\end{proof}

Before proving our hierarchy theorem for $\Ott$, we prove a simple lemma that bounds from above the size of a circuit family computing languages in $\Ott$.

\begin{lemma}\label{lem:O2_UB}
    $\Ott[t(n)] \subseteq \SIZE[t(n)\log(t(n))]$. 
\end{lemma}
\begin{proof}
    Consider any language $L \in \Ott[t(n)]$, and let $V(\cdot, \cdot, \cdot)$ be its $t(n)$-time verifier. For any integer $n \in \N$, let $y_n, z_n \in \bit^{t(n)}$ be the irrefutable proofs for input size $n$. By \cref{lem:tm_to_ckt} we can convert $V(\cdot, \cdot, \cdot)$ into a circuit family $\{C_{n}\} \subseteq \cc{SIZE}[t(n)\log (t(n)]$. The values $y_n$ and $z_n$ can be hard-coded into $C_{n}$, and hence this circuit will decide $L$ on all inputs of size $n$.
\end{proof}

Having both an upper bound on the size of circuits simulating an $\Ott$ computation, and also a lower bound for $\Ott$ against circuits, we can use the circuit size hierarchy (\cref{thm:ckt_size_hierarchy}) to define a time hierarchy on $\Ott$.

\begin{theorem}\label{thm:O2_heirarchy}
        For $n \in \N$, let $t: \N \rightarrow \N$ be a time constructible function, s.t. $t(n) > n \geq 12$ then: $\Ott[t(n)] \subsetneq \Ott[t(n)^4\log^9(t(n))]$.
\end{theorem}
\begin{proof}
    Combining \cref{thm:O2_LB}, \cref{lem:O2_UB}, and Circuit Size Hierarchy (\cref{thm:ckt_size_hierarchy}) we have: \\
     \[\Ott[t(n)] \subseteq \cc{SIZE}[t(n) \log t(n)] \subsetneq \cc{SIZE}[t(n) \log^{\frac{5}{4}} t(n)] \; ,\]
     and \\
      \[ \Ott[t(n)^4 \log^9(t(n))] \not\subseteq \cc{SIZE}[t(n) \log^{\frac{5}{4}} t(n)] \; .\]
\end{proof}

\begin{theorem}\label{thm:O2_heirarchy_fine_grained}
    For $n \in \N$, let $t: \N \rightarrow \N$ be a time constructible function, s.t. $t(n) \geq n$ then: for all $\eps > 0, \Ott[t(n)] \subsetneq \Ott[t(n)^{1 + \eps}]$.
\end{theorem}

\begin{proof}
    Let us assume that $\Ott[t(n)^{1 + \eps}] \subseteq \Ott[t(n)]$, then by translation we have that:
    
    $$\Ott[t(n)^{{(1+\eps)}^2}] \subseteq \Ott[t(n)^{1 + \eps}] \subseteq \Ott[t(n)]$$ Inducting on any $k > 2$, we get that $\Ott[t(n)^{{(1 + \eps)}^k}] \subseteq \Ott[t(n)]$. Now setting $k = \ceil{3.1/ \eps}$, by Bernoulli's inequality we have that ${(1 + \eps)}^k \geq {(1 + \eps)}^{3.1/\eps} \geq 4.1$. Therefore, $\Ott[t(n)^{4.1}] \subseteq \Ott[t(n)]$, which contradicts \cref{thm:O2_heirarchy}.
\end{proof}

\begin{remark}
    We note here that while we are able to achieve a better gap in our hierarchy theorem in \cref{thm:O2_heirarchy_fine_grained} over \cref{thm:O2_heirarchy}, there is  a trade off. The hierarchy theorem defined in \cref{thm:O2_heirarchy} is explicit, that is we have an explicit language not known to be contained in the smaller class. However, when we apply our translation argument to get better parameters in \cref{thm:O2_heirarchy_fine_grained} we lose this explicitness.
\end{remark}

\section{Sparsity}

In this section, we use \emph{sparse extensions} to get various structural complexity results. We prove a more fine-grained statement of \cref{THM:main3} which states that any language in $\Stt[t(n)]$ with a \emph{uniformly-sparse extension} is actually in $\Ott[t(n)^2]$. This lets us extract as a corollary another proof of $\Se = \Oe$. As another application of \emph{sparse extensions}, we are able to recover the fixed polynomial lowerbounds for $\Ot$ from the previous section as stated in \cref{THM:main1}. Finally we show connections between \emph{sparse extensions} and open problems posed by \cite{GoldreichMeir2015}.

\begin{lemma}\label{lem:Ot_simul_St}
    Let $L \in \Stt[t(n)]$. If $L$ has a $t(n)$-uniformly-sparse extension then $L \in \Ott[t(n)^2]$.
\end{lemma}

\begin{proof}
For any $n$, let $L'$ be the \emph{$t(n)$-uniformly-sparse extension} of $L$, and let $\mathcal{F}$ be the $\cc{TIME}[t(n)]$ predicate that decides membership in $L'$. We will now design an $\Ott[t(n)^2]$ verifier $V$ for $L \restrict{n}$. Since both $L$ and $L'$ are $t(n)$-$\cc{SPARSE}$, we have that for most $x \in \bit^n$: $L \restrict{n}(x) = L' \restrict{n}(x) = 0$. $V$ will first use $\mathcal{F}$ to efficiently filter out most non-membership in $L' \restrict{n}$, and hence $L \restrict{n}$ in $\cc{TIME}[t(n)]$. Now $V$ only has to decide membership in $L$ over $t(n)$ many inputs $X = \{x \in \{0,1\}^n : \mathcal{F}(x) = 1 \}$. We will use the fact that since $L \in \Stt[t(n)]$, for all $x \in X$, if $x \in L$ there is an irrefutable YES certificate $y_x$ and if $x \notin L$ there is an irrefutable NO certificate $z_x$ and a verifier $V^*$, running in $\cc{TIME}[t(n)]$ s.t.
    \begin{itemize}
        \item if $x \in L$, $\exists y_x$, $\forall z$ s.t. $V^*(x,y_x,z) = 1$
        \item if $x \notin L$, $\exists z_x$, $\forall y$ s.t. $V^*(x,y,z_x) = 0$
    \end{itemize}

    \begin{figure}[!htbp]

    \begin{wbox}
    $V(x, Y^*, Z^*):$
    \begin{enumerate}[(1)]
    \item Set output $ = 1$.

    \item If $\mathcal{F}(x) = 0$, return $0$.

    \item Parse $Y^*$ to get $y_x^*$.

    \item For $z_i \in Z^*$, do:

        \begin{enumerate} [(a)]
            \item $\text{output} = \text{output} \wedge V^*(x, y_x^*, z_i)$.
        \end{enumerate}

    \item Return output.    
    
\end{enumerate}
    \end{wbox}
\caption{$\Ott[t(n)^2]$ Verifier for Language in $\Stt[t(n)]$ with $t(n)$-\emph{uniformly-sparse extension}}
\end{figure}

    Consider the string $Y^{*}$ which encodes a table of YES witnesses $y_x^*$ for every input $x \in X$. When $x \in L$ we set $y_x^* = y_x$, and when $x \notin L$ we will set $y_x^* = 0^{t(n)}$. The size of $Y^{*}$ is $O(t(n)^2)$, since there are at most $t(n)$ entries in the table each of length $t(n) + n$.
    For every $x \in X \cap \overline{L}$, let $z_x$ be the irrefutable NO-certificate corresponding to $x$ for $V^*$. We set $Z^*$ to be the concatenation of all such $z_x$. The size of $Z^*$ is also at most $t(n)^2$.
     
     We now show that $Y^{*}$ and $Z^{*}$ will serve as oblivious irrefutable \say{YES} and \say{NO} certificates respectively for $V$. On input $(x, Y^*, Z^*)$, $V$ first parses $Y^*$ to find the corresponding $y_x^*$ in time $\cc{TIME}[t(n)^2]$. Then for each $z_i \in Z^*$ we run $V^*(x, y_x^*, z_i)$. If for all $z_i$, $V^*(x, y_x^*, z_i) = 1$ then $V$ outputs $1$, otherwise $V$ will output $0$. Since we are making at most $t(n)$ calls that each cost $\cc{TIME}[t(n)]$, $V$ runs in $\cc{TIME}[t(n)^2]$.
     
    To see correctness, we first analyze the case when $x \in L$, then by construction $Y^{*}$ includes $y_x^* = y_x$ and $V$ will output $1$. On the other hand if $x \notin L$ then there is an irrefutable no-certificate $z_x$ in $Z^{*}$ so there is no $y_i$ such that $V(x, y_i, z_x) = 1$. Hence $V$ outputs $0$.
    
\end{proof}

By taking $t(n)$ to be a polynomial in \cref{lem:Ot_simul_St} we directly get \cref{cor:OP_SP} (also \cref{THM:main3}) relating $\Ot$ and $\St$. 

\begin{corollary}\label{cor:OP_SP}
    If $L \in \St$ and $L$ has an uniformly-sparse extension, then $L \in \Ot$
\end{corollary}

Similarly, one can prove \cref{THM:main4} by showing the same consequence for $\coNP$ vs $\coONP$, thus making a partial progress towards the open questions posed by Goldreich and Meir in \cite{GoldreichMeir2015}. In the exponential regime, since all languages have the trivial $2^n$-\emph{uniformly-sparse extension} we get the equivalence between $\Oe$ and $\Se$ as seen in \cref{cor:OE_SE}.

\begin{corollary}\label{cor:OE_SE}
    $\Se = \Oe$
\end{corollary}

\begin{proof}
    As noted in \cref{sec:prelim_sparsity}, every language is $2^{n}$-$\cc{SPARSE}$, and has the trivial $2^n$-\emph{uniformly-sparse extension}: $\bit^*$. When $t(n) = 2^n$, by \cref{lem:Ot_simul_St} we get that $\Stt[2^n] \subseteq \Ott[2^{2n}]$.
\end{proof}

In particular, the following lemma shows that the hard language in $\Stt[t(n)]$ defined in \cref{thm:S2_LB} admits a $t(n)$-\emph{uniformly-sparse extension}, giving another proof of \cref{cor:fixed_poly_direct}.

\begin{lemma}\label{lem:S2_LB_sparse_extension}
For $n \in \N$, let $t:\N \rightarrow \N$ be a time-constructible function, s.t. $t(n) > n \geq 12$ then, there is an explicit language $L_t \in \Stt[t(n)]$ s.t. $L_t \notin \SIZE\left[\frac{t(n)^{1/4}}{\log(t(n))}\right]$. Moreover, $L_t$ has a $t(n)$-uniformly-sparse extension $L'_t$.
\end{lemma}

\begin{proof}
    Let $L_t$ be the $\Stt[t(n)]$ language defined in the proof \cref{thm:S2_LB} with the characteristic string $\mathcal{X}_{L_t\restrict{n}} := f_n || 0^{2^n -   |E_{n,s}| - 1}$. We now define the language $L_t'$ whose characteristic string $\mathcal{X}_{L_t'\restrict{n}} := 1^{|E_{n,s}|+ 1} || 0^{2^n - |E_{n,s}| - 1}$. To see that this $L_t'$ is a $t(n)$-$\emph{uniformly-sparse extension}$ of $L_t$, clearly $L_t \subseteq L_t'$. Moreover membership of $x \in L_t'$ can be decided by checking if the binary value of $x$ is less than or equal to  $|E_{n,s}|+1$ which can be done in $\cc{TIME}[n] \subseteq \cc{TIME}[t(n)]$.
\end{proof}

Equipped with this lemma we have an alternative proof of fixed polynomial lower bounds for $\Ot$ as stated in \cref{THM:main1}.

\begin{corollary} \textnormal{(\cref{THM:main1})}
    For every $k \in \N$, $\Ot \not\subseteq \cc{SIZE}[n^k]$. Moreover, for every $k$ there is an explicit language $L_k$ in $\Ot$ s.t. $L_k \notin \cc{SIZE}[n^k]$.
\end{corollary}

\begin{proof}
    Fix $t(n) = n^{5k}$. Then by \cref{lem:S2_LB_sparse_extension} there is an explicit language $L_k$ such that $L_k \not\subseteq \cc{SIZE}[n^{k}]$, and $L_k$ has an \emph{uniformly-sparse extension}. Applying \cref{lem:Ot_simul_St} we have that $L_k \in   \Ott[n^{10k}] \subseteq \Ot$.      
\end{proof}

\section{Extension}

In this section we show how to extend our result to even, potentially, smaller class $
\Ot \cap \Lt$, where $\Lt$ is a complexity classes defined by Korten and Pitassi in \cite{KP24}.

\begin{definition}[\cite{KP24}]
A language $L$ belongs to $\Lt$, if there is a polynomial-time predicate
$R(x,y,z)$ and a polynomial $p$, so that for all $x: R(x, \cdot , \cdot)$ defines a total order on
$\set{0, 1}^ {p(n)}$ whose minimal element $m$ satisfies $m_1 = L(x)$.
That is, for every $x$, the first bit of the minimal elements denotes whether $x \in L$. 
\end{definition}

It follows directly from the definition that $\Lt \subseteq \St$. Moreover, it has been also shown that $\Lt$ can equivalently be characterized as the class of all languages that are many-to-one reducible, to the so-called \emph{Linear Ordering Principle} ($\cc{LOP}$), in polynomial time. At the same time, $\cc{Avoid}$ also reduces to $\cc{LOP}$. Consequently, one can formulate the following analogue of \cref{Alg:FS2P}.

\begin{theorem}[\cite{KP24}]
\label{Alg:L2P}
$\cc{Avoid}$ is many-to-reducible to $\Lt$, in polynomial time.    
\end{theorem}

Given that, one can repeat the argument of \cref{thm:S2_LB}, invoking \cref{Alg:L2P} instead of \cref{Alg:FS2P}, to obtain the following analogue of \cref{cor:fixed_poly_direct}.

\begin{corollary}
\label{cor:fixed_poly_direct_L2p}
 For all $k \in \N$, there exists an explicit language $A_k \in \Lt$ s.t. $A_k \not\subseteq \cc{SIZE}[n^k]$. Moreover, $A_k$ has an uniformly-sparse extension.
\end{corollary}

Finally, as $\Lt \subseteq \St$ the explicit ``hard'' languages actually collapse all the way to $\Lt \cap \Ot$.

\begin{corollary}
\label{cor:fixed_poly_direct_L2p_O2p}
 For all $k \in \N$, there exists an explicit language $A_k \in \Lt \cap \Ot$ s.t. $A_k \not\subseteq \cc{SIZE}[n^k]$. 
\end{corollary}

\begin{proof}
As $\Lt \subseteq \St$ and  for all $k \in \N: A_k$ has an uniformly-sparse extension,
by  \cref{cor:OP_SP} $A_k \in \Ot$. 
\end{proof}
\section{Open Problems}

We conclude with a few interesting open problems:

\begin{itemize}
    \item Can we show that every sparse $\St$ language is also in $\Ot$? 
    \item Can we tighten the gap in the $\Ott$ hierarchy theorem (\cref{THM:main2})? 
    \item Can we show a non-trivial upper bound for $\Ot$, for example $\P^{\NP}, \MA, \PP$? This would imply explicit fixed-polynomial lower bounds for such classes. On the other hand, we do note that under reasonable derandomization assumptions, $\Ot \subseteq \St = \P^{\NP}$.
    \item Can we arrive at something interesting about time hierarchy theorem for semantic classes where fixed-polynomial lower bounds are known e.g. $\St, \; \ZPP^{\NP}$, assuming $\NP \not\subseteq \PPoly$? For instance, if $\NP \subseteq \PPoly$, then it follows that $\St\subseteq \PPoly$. One could then invoke the circuit size hierarchy theorem (\cref{thm:ckt_size_hierarchy}) to establish a hierarchy theorem for $\Stt$, similar to how we obtain the hierarchy theorem for $\Ott$.

\end{itemize}

\section{Acknowledgements}

The authors would like to thank Alexander Golovnev, William Gasarch and
Edward Hirsch for many helpful discussions and feedback on an earlier draft of this manuscript. 
The authors would also like to thank the anonymous referees for useful comments.

\bibliographystyle{alpha}
\bibliography{refs,bibliography}

\end{document}